\documentclass[onefignum,onetabnum]{siamart190516}

\usepackage{lipsum}
\usepackage{amsfonts}
\usepackage{graphicx}
\usepackage{epstopdf}
\usepackage[noend]{algpseudocode}
\ifpdf
  \DeclareGraphicsExtensions{.eps,.pdf,.png,.jpg}
\else
  \DeclareGraphicsExtensions{.eps}
\fi
\usepackage{subcaption}

\newcommand{\Hypergeometric}[5]{\ensuremath{{}_{#1}  F_{#2}\mathord{\left[\begin{smallmatrix}#3 \\ #4\end{smallmatrix};#5\right]}}}

\newsiamremark{remark}{Remark}
\newsiamremark{hypothesis}{Hypothesis}
\crefname{hypothesis}{Hypothesis}{Hypotheses}
\newsiamthm{claim}{Claim}

\headers{On the concentration of the maximum degree\ldots}{A. Frieze, K. Turowski, and W. Szpankowski}

\title{On the concentration of the maximum degree in~the~duplication-divergence models\thanks{Submitted to the editors DATE.
\newline
The preliminary, partial, and weaker versions of these results appeared in WG 2020 \cite{frieze-wg} and COCOON 2021 \cite{frieze-cocoon} conference proceedings.
\funding{This work was funded in part by NSF Center for Science of Information Grant CCF-0939370, and NSF Grants DMS1661063 and CCF-2006440. \newline
Krzysztof Turowski's research was funded by the Polish National Science Center $2020$/$39$/D/ST$6$/$00419$ grant. For the purpose of Open Access, the author has applied a CC-BY public copyright license to any Author Accepted Manuscript (AAM) version arising from this submission. 
Wojciech Szpankowski's work was funded in part by the Polish National Science Center $2018$/$31$/B/ST$6$/$01294$ grant.}}}

\author{Alan Frieze\thanks{Department of Mathematical Sciences, Carnegie Mellon University, Pittsburgh, PA, USA
  (\email{alan@random.math.cmu.edu}).}
\and Krzysztof Turowski\thanks{Theoretical Computer Science Department, Jagiellonian University, Krak\'ow, Poland 
  (\email{krzysztof.szymon.turowski@gmail.com}).}
\and Wojciech Szpankowski\thanks{Center for Science of Information, Department of Computer Science, Purdue University, West Lafayette, IN, USA 
  (\email{spa@cs.purdue.edu}).}}

\usepackage{amsopn}

\usepackage{amssymb,float,marginnote}
\usepackage{cleveref}
\usepackage{epsfig}
\usepackage{color, lipsum}
\usepackage{booktabs}
\usepackage{tikz, tikz-qtree}
\usepackage{multicol, multirow}
\usepackage[backgroundcolor=gray!20,linecolor=black]{todonotes}

\newcommand{\E}{\mathbb{E}}
\newcommand{\Var}{\mathrm{Var}}

\newcommand{\parent}[1]{\ensuremath{\textrm{parent}({#1})}}
\newcommand{\DD}{\texttt{DD}}

\newcommand{\A}{\mathcal{A}}
\newcommand{\B}{\mathcal{B}}
\newcommand{\C}{\mathcal{C}}

\ifpdf
\hypersetup{
  pdftitle={On the concentration of the maximum degree in the duplication-divergence models},
  pdfauthor={A. Frieze, K. Turowski, and W. Szpankowski}
}
\fi

\begin{document}

\maketitle


\begin{abstract}
We present a rigorous and precise analysis of the maximum degree and the average degree in a dynamic duplication-divergence graph model introduced by Sol\'e, Pastor-Satorras et al. in which the graph grows according to a duplication-divergence mechanism, i.e. by iteratively creating a copy of some node and then randomly alternating the neighborhood of a new node with probability $p$.
This model captures the growth of some real-world processes e.g.~biological or social networks.

In this paper, we prove that for some $0 < p < 1$ the maximum degree and the average degree of a duplication-divergence graph on $t$ vertices are asymptotically concentrated with high probability around $t^p$ and $\max\{t^{2 p - 1}, 1\}$, respectively, i.e. they are within at most a polylogarithmic factor from these values with probability at least $1 - t^{-A}$ for any constant $A > 0$.
\end{abstract}



\section{Introduction}

Studying properties of random graphs is a~popular topic of research in computer science and discrete mathematics since the seminal work of Paul Erd\H{o}s and Alfr{\'e}d R\'enyi \cite{erdos-renyi}.
This model was studied extensively using various probabilistic and analytic methods.
The research mostly concentrated on a few broad topics: distribution of structural properties of graphs (e.g. the number of edges, degrees of fixed vertex, maximum degree, diameter), the existence of special subgraphs (e.g. motif counting, longest paths, maximum matching, Hamilton cycles), values of well-known combinatorial parameters (e.g. largest independent set, chromatic number), or extremal properties (Ramsey- and Tur\'an-type) -- see e.g. surveys of results in \cite{bollobas-random,frieze-introduction,luczak-random,hofstad}.

The widening array of application domains ranging from biology to finance to social science inspired further directions of research: first, there appeared an idea to bring the models to the real-world data and to study important aspects, such as centrality, degree correlation, community detection, or graph compression \cite{kaminski-pralat,latora-nicosia,networks-overview}.
Second, more models of random networks were developed e.g. for inhomogeneous random graphs, geometric random graphs, preferential attachment graphs, or duplication graphs \cite{chung-book,frieze-introduction,hofstad}.
Often, these models were inspired by some generation mechanisms (e.g. rich-get-richer), or properties (e.g. scale-free/power-law property) that were claimed at work for the real-world networks \cite{faloutsos}.

In particular, since the late 1990s attention turned toward dynamics graphs in which
the behavior of networks evolves in time, e.g. when sets of vertices and/or edges are functions of time, which is definitely the case for certain biological (e.g. protein-protein networks) and social networks (e.g. graph of citations).

One of the family of such networks is the so-called duplication models \cite{chung-duplication,chung-book}.
It was observed that the evolutionary dynamics of protein interaction networks can be described by simple duplication and mutation rules \cite{ohno,zhang-evolution}.
For example, the main mechanisms in such models are duplication and divergence:
when vertices arrive one by one, they are created as copies of some already existing node, chosen uniformly at random (duplication), and then the neighborhood is typically altered randomly according to some predefined rules (divergence).

In this paper we study a particular duplication-divergence model, first introduced by Sol\'e, Pastor-Satorras et al. \cite{sole-model}.
This model is a promising object of inquiry since it has been shown empirically that its degree distribution, small subgraph (graphlets) counts, and the number of symmetries fit very well the structure of some real-world biological and social networks, e.g. protein-protein and citation networks. More precisely, there exist heuristics to infer the underlying parameters of the model from various biological networks, which enable us to generate similar graphs in terms of degree distributions, k-hop reachability, closeness, betweenness, and graphlet frequency \cite{hormozdiari-seed,li-mle} (see also an alternative method of parameter estimation in \cite{turowski-kdd}).
It also turns out that this model often outperformed alternative ones in terms of systematically replicating the degree distribution, small subgraph (graphlets) counts, and symmetries of the input networks \cite{colak-dense,shao-choosing,turowski-kdd}.
This suggests a possible real-world significance for the duplication-divergence model, which further motivates the studies of its structural properties.

However, it is also one of the least understood models, much less so than the Erd\H{o}s-R\'enyi or preferential attachment models.
At the moment there exist only a handful of precise results related to the behavior of the degree distribution of the graphs generated by this model.
Our contribution is a step towards closing this gap.
In short, we prove an asymptotic tight concentration of two parameters in duplication-divergence graphs: maximum degree, and average degree (or, equivalently, the number of edges) around their mean values.

The paper is organized as follows: in \Cref{sec:model} we define formally the duplication-divergence model, and we present an overview of the previous results related to the properties of the degree distribution.
Then, in \Cref{sec:maximum} we introduce our result for the maximum degree, with proof split into three parts: 
in \Cref{sec:maximum-upper-early} and \Cref{sec:maximum-upper-late} we prove upper bounds for the degrees of the earliest and later vertices arriving in the graph, respectively, and in \Cref{sec:maximum-lower} we give a~proof of the lower bound for the degree of the first vertices, which is effectively also the lower degree of the maximum degree.
Next, we proceed with \Cref{sec:average}, containing the proofs of the upper and the lower bounds for the average degree (or, equivalently, the total number of edges in the graph), respectively.
Finally, we offer some further problems and hypotheses that stem from our current research.

This work is a substantial extension of two conference papers: the one presented at COCOON 2021 \cite{frieze-cocoon} which contained the weaker concentration results for maximum degree only for the case $\frac{1}{2} < p < 1$, and the one presented at WG 2020 \cite{frieze-wg} which contained the weaker claims (proved using different methods) for average degree and for degrees only of the earliest vertices.

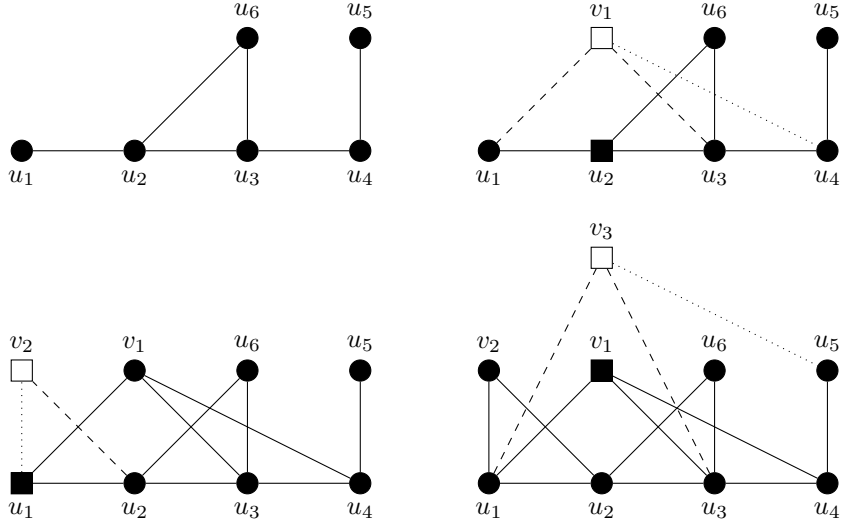
\begin{figure}[htbp]
  \tikzstyle{every node}=[draw, circle, fill=black, inner sep=0pt, minimum width=8pt]
  \centering
  \begin{subfigure}[b]{0.45\textwidth}
    \centering
    \begin{tikzpicture}[scale=1.5]
      \node (U1) at (-2, 0) [label=below:$u_1$]{};
      \node (U2) at (-1, 0) [label=below:$u_2$]{};
      \node (U3) at (0, 0) [label=below:$u_3$]{};
      \node (U4) at (1, 0) [label=below:$u_4$]{};
      \node (U5) at (1, 1) [label=$u_5$]{};
      \node (U6) at (0, 1) [label=$u_6$]{};
      \draw (U1) -- (U2) -- (U3) -- (U4) -- (U5);
      \draw (U2) -- (U6) -- (U3);
    \end{tikzpicture}
  \end{subfigure}
  \quad
  \begin{subfigure}[b]{0.45\textwidth}
    \centering
    \begin{tikzpicture}[scale=1.5]
      \node (U1) at (-2, 0) [label=below:$u_1$]{};
      \node[rectangle,minimum size=8pt] (U2) at (-1, 0) [label=below:$u_2$]{};
      \node (U3) at (0, 0) [label=below:$u_3$]{};
      \node (U4) at (1, 0) [label=below:$u_4$]{};
      \node (U5) at (1, 1) [label=$u_5$]{};
      \node (U6) at (0, 1) [label=$u_6$]{};
      \node[rectangle,minimum size=8pt,fill=white] (V1) at (-1, 1) [label=$v_1$]{};
      \draw (U1) -- (U2) -- (U3) -- (U4) -- (U5);
      \draw (U2) -- (U6) -- (U3);
      \draw[dashed] (U3) -- (V1) -- (U1);
      \draw[dotted] (V1) -- (U4);
    \end{tikzpicture}
  \end{subfigure}
  \\\vspace{8pt}
  \begin{subfigure}[b]{0.45\textwidth}
    \centering
    \begin{tikzpicture}[scale=1.5]
      \node (U2) at (-1, 0) [label=below:$u_2$]{};
      \node (U3) at (0, 0) [label=below:$u_3$]{};
      \node (U4) at (1, 0) [label=below:$u_4$]{};
      \node (U5) at (1, 1) [label=$u_5$]{};
      \node (U6) at (0, 1) [label=$u_6$]{};
      \node (V1) at (-1, 1) [label=$v_1$]{};
      \node[rectangle,minimum size=8pt] (U1) at (-2, 0) [label=below:$u_1$]{};
      \node[rectangle,minimum size=8pt,fill=white] (V2) at (-2, 1) [label=$v_2$]{};
      \draw (U1) -- (U2) -- (U3) -- (U4) -- (U5);
      \draw (U2) -- (U6) -- (U3);
      \draw (U3) -- (V1) -- (U1);
      \draw (V1) -- (U4);
      \draw[dashed] (U2) -- (V2);
      \draw[dotted] (U1) -- (V2);
    \end{tikzpicture}
  \end{subfigure}
  \quad
  \begin{subfigure}[b]{0.45\textwidth}
    \centering
    \begin{tikzpicture}[scale=1.5]
      \node (U1) at (-2, 0) [label=below:$u_1$]{};
      \node (U2) at (-1, 0) [label=below:$u_2$]{};
      \node (U3) at (0, 0) [label=below:$u_3$]{};
      \node (U4) at (1, 0) [label=below:$u_4$]{};
      \node (U5) at (1, 1) [label=$u_5$]{};
      \node (U6) at (0, 1) [label=$u_6$]{};
      \node[rectangle,minimum size=8pt] (V1) at (-1, 1) [label=$v_1$]{};
      \node (V2) at (-2, 1) [label=$v_2$]{};
      \node[rectangle,minimum size=8pt,fill=white] (V3) at (-1, 2) [label=$v_3$]{};
      \draw (U1) -- (U2) -- (U3) -- (U4) -- (U5);
      \draw (U2) -- (U6) -- (U3);
      \draw (U3) -- (V1) -- (U1);
      \draw (V1) -- (U4);
      \draw (U2) -- (V2);
      \draw (U1) -- (V2);
      \draw[dashed] (U1) -- (V3) -- (U3);
      \draw[dotted] (V3) -- (U5);
    \end{tikzpicture}
  \end{subfigure}
  \caption{Graph evolution in the duplication-divergence model: new vertices and their parents are marked as white and black squares, respectively; $p$-edges and $r$-edges are denoted by dashed and dotted lines.}
  \label{fig:example-2}
\end{figure}

\section{Model definition and earlier work}
\label{sec:model}

Throughout the paper we use standard graph notation from \cite{diestel}, e.g. $V(G)$ denotes the vertex set of a graph $G$, $\deg_G(s)$ is the degree of node $s$ in $G$, and we write $\Delta(G)$ and $D(G)$ for the maximum degree and the average degree in $G$. Let also $N_G(s)$ denote the open neighborhood of $s$ in $G$.
All graphs considered in the paper are simple, i.e. without loops or multiple edges.

Additionally, since we are eventually dealing with a probability space over graphs on $t$ vertices, let $G_t$ denote a random variable representing a graph on $t$ vertices.
Finally, since we are dealing with graphs growing sequentially, we assume that the vertices are identified with the natural numbers according to their arrival time.
For simplicity, we introduce the notation $\deg_t(s)$ for the random variable denoting the degree of vertex $s$ in $G_t$.
Clearly, $\Delta(G_t)$ and $D(G_t)$ are random variables denoting the maximum degree and the average degree in $G_t$.

Let us now formally define the duplication-divergence model, denoted $\DD(t, p, r)$, introduced by Sol\'e et al. \cite{sole-model,pastor-evolving}.
Let $G_{t_0}$ be some graph on $t_0 \le t$ vertices, with vertices having distinct labels from $1$ to $t_0$.
Now, for every $i = t_0, t_0 + 1, \ldots, t - 1$ we create $G_{i + 1}$ from $G_i$ 
according to the following rules:
\begin{enumerate}
  \item we add a new vertex with label $i + 1$ to the graph,
  \item we choose a vertex $u$ from $G_i$ uniformly at random -- and we denote $u$ as $\parent{i + 1}$,
  \item for every vertex $v$:
  \begin{enumerate}
    \item if $v$ is adjacent to $u$ ($v \in N_{G_i}(u)$) in $G_i$, then add an edge between $v$ and $i + 1$ with probability $p$,
    \item if $v$ is not adjacent to $u$ in $G_i$ ($v \notin N_{G_i}(u)$), then add an edge between $v$ and $i + 1$ with probability $\frac{r}{i}$. Note that this case also occurs when $v = u$, since $u \notin N_{G_i}(u)$.
  \end{enumerate}
\end{enumerate}
All edge additions are independent Bernoulli random variables.

Since both $p$ and $\frac{r}{i}$ for $i = t_0, \ldots, t - 1$ are probabilities, we allow the parameter space to be $p \in [0, 1]$ and $r \in [0, t_0]$.

There is indeed a duplication-divergence mechanism at work since we can think of the equivalent set of rules in the form ``copy a vertex from $G_i$ uniformly at random'', ``remove its neighbors independently at random with probability $1 - p$'', and ``add edges to all other vertices independently at random with probability $\frac{r}{i}$''.

Throughout the paper we will refer to the standard Big-O Landau notation, as popularized e.g. in \cite{knuth}. Let us recall its basic notion: $f(n) = O(g(n))$ for some functions $f$ and $g$ such that $\exists_{k > 0} \exists_{n_0} \forall_{n > n_0} |f(n)| \le |k \cdot g(n)|$.
Additionally, we will use
\begin{itemize}
    \item $f(n) = \Omega(g(n))$ when $g(n) = O(f(n))$,
    \item $f(n) = \Theta(g(n))$ when both $f(n) = O(g(n))$ and $g(n) = O(f(n))$,
    \item $f(n) = o(g(n))$ when $f(n) = O(g(n))$ but not $f(n) = O(g(n))$.
\end{itemize}
Intuitively, $f(n) = \Omega(g(n))$ when $\lim_{n \to \infty} \frac{|f(n)|}{|g(n)|} \in [k_1, k_2]$ for some $0 < k_1 < k_2$.
Since in the model both $p$ and $r$ (and the order of initial graph $G_{t_0}$) are constants, the asymptotic results are given exclusively in terms of $t$.

As it was mentioned earlier, there are only a few rigorous results for the $\DD(t, p, r)$ model and its special cases.
For $0 < p < 1$ and $r = 0$, it was proved in \cite{hermann} that asymptotically there exists a phase transition for the limiting distribution of degree frequencies: if $p \le p^*$, then almost all vertices are isolated, i.e. the number of non-isolated vertices in $G_t$ is $o(t)$, and if $p > p^*$, then only a constant fraction of vertices (with an explicit constant) are isolated.
Moreover, it was proved that for any $k$ the fraction of vertices of degree $k$ in $G_t$ converges to $0$, and therefore there is no limiting degree distribution for $p > p^*$.
From \cite{li-degree} it is known that the number of vertices of degree one in $G_t$
is $\Omega(\log{t})$ but again the precise rate of growth of the number of vertices
with any fixed degree $k > 0$ is currently unknown.

However, also for the same case in \cite{jordan,jacquet-aofa} it was shown for $p < \exp(-1)$ that the (only) connected component in $G_t$ exhibits a power-law property with the scale parameter $\gamma$ which is the solution of $3 = \gamma + p^{\gamma - 2}$.

For the general case, the two main parameters under consideration were the degree of fixed vertices $\deg_t(s)$ and the average degree of $G_t$ defined as
\begin{align*}
    D(G_t)= \frac{1}{t} \sum_{s=1}^t \deg_t(s).
\end{align*}

It was shown in \cite{turowski-expected} that we can solve the recurrence equation for the expected average degree and obtain
\begin{theorem}
\label{thm:exp_avg_deg}
For $t \to \infty$ it holds that
\begin{align*}
  \E[D(G_t)] & =
  \begin{cases}
    \Theta(1) & \text{if $p < \frac{1}{2}$ and $r > 0$,} \\
    \Theta(\ln{t}) & \text{if $p = \frac{1}{2}$ and $r > 0$,} \\
    \Theta(t^{2 p - 1}) & \text{otherwise.}
  \end{cases}
\end{align*}
\end{theorem}
In a similar fashion it was shown that the expected degree of a vertex $s$ is given by the following theorem:
\begin{theorem}
\label{thm:exp_given_deg}
For $t \to \infty$, it holds that
\begin{align*}
  \E[\deg_{t}(s)] & =
  \begin{cases}
    \Theta\left(\log{\left(\frac{t}{s}\right)}\right) & \text{if $p = 0$ and $r > 0$,} \\
    \Theta\left(\left(\frac{t}{s}\right)^p\right) & \text{if $0 < p < \frac{1}{2}$ and $r > 0$,} \\
    \Theta\left(\sqrt{\frac{t}{s}} \log{s}\right) & \text{if $p = \frac{1}{2}$ and $r > 0$,} \\
    \Theta\left(\left(\frac{t}{s}\right)^p s^{2 p - 1}\right) & \text{otherwise.}
  \end{cases}
\end{align*}
\end{theorem}
Clearly, the latter result for the earliest vertices implies that the expected maximum degree is $\Omega(t^p)$ for all $0 < p < 1$.

In fact, in \cite{turowski-expected} the authors obtained more than  just \Cref{thm:exp_avg_deg} and \Cref{thm:exp_given_deg}, because they derived the exact formulae for both $\E[D(G_t)]$ and $\E[\deg_{t}(s)]$ with their very convoluted leading coefficients (depending on $s$, $p$, $r$) together with the asymptotics for $\Var[D(G_t)]$ and $\Var[\deg_{t}(s)]$.

The natural question then is to show that these random variables are concentrated, i.e. whether by moving only some small (e.g. polylogarithmic) factor from the mean we could observe the polynomial tail decay.
Intuitively, for the later vertices we should not expect such a phenomenon: since the parent of a new vertex is drawn uniformly, and there are two binomial processes on top of it, we expect the degree distribution of $\deg_t(t)$ rather reflect the whole degree distribution, which for some cases we know (and for all other we stipulate, based on simulations) is not concentrated.
However, as we will see in the next sections for the maximum degree and the average degree we can answer this question in the affirmative.

\section{Maximum degree}
\label{sec:maximum}

In this section we present our main result concerning the concentration of the maximum degree $\Delta(G_t)$. We formulate it in the next theorem.

\begin{theorem}
Let $0 < p < 1$. Asymptotically for $G_t \sim \DD(t, p, r)$
\begin{align*}
  \Pr&[(1 - \alpha) t^p \le \Delta(G_t) \le (1 + \alpha) t^p \log^{2 - p^2}(t)] = 1 - O(t^{-A})
\end{align*}
for any constants $\alpha > 0$ and $A > 0$.
\end{theorem}

We prove separately a lower bound and a matching (within a polylogarithmic factor) upper bound.
The main idea of the upper bound proof, presented in the next subsection, is as follows:
we first in \Cref{def:sequences} introduce auxiliary deterministic sequences $(t_i)_{i = 0}^k$ and $(X_{t_i})_{i = 0}^k$ such that $t_0 < \ldots < t_{k - 1} < t \le t_k$.
Although at first glance the dependency between $(t_i)_{i = 0}^k$ and $(X_{t_i})_{i = 0}^k$ given in this definition could seem very convoluted, the intuition behind it is very simple: by doing this we can prove with little effort that $X_{t_i}$ grows close to $t_i^p$, provided that we choose the right parameters.
Indeed, we show that $X_t \le (1 + \alpha) t^p \log^{2 - p^2}(t)$ for any constant $\alpha > 0$.

This way, we want $(X_{t_i})_{i = 0}^k$ to be a good (i.e. holding with high probability) upper bound for $\deg_{t_i}(s)$ for all $i = 0, \ldots, k$ and all $s \le t_0$ (denoted as \emph{early vertices}), which in turn should give us a similar lower bound $\deg_t(s)$ in terms of $X_t$ whp.
We proceed in two major steps: first, by construction, we have $\deg_{t_0}(s) \le t_0 = X_{t_0}$, and second, we prove a bound on $\deg_{t_{i + 1}}(s) - \deg_{t_i}(s)$ that ensures it does not exceed $X_{t_{i + 1}} - X_{t_i}$ with high probability. The latter part is achieved by providing an adequate upper bounding of $\deg_{t_{i + 1}}(s) - \deg_{t_i}(s)$ by a sum of independent Bernoulli variables, so the Chernoff bound can be employed -- and by applying a telescoping sum we establish that $\deg_t(s) \le X_t$ with high probability for all $s \le t_0$.
Therefore, we find for early vertices $s$ (i.e. $s \le t_0$) a Chernoff-type bound on the growth of $\deg_\tau(s)$ over an interval of certain length $h$.

The second part of the proof of our upper bound on the maximum degree is inductive: we prove that with high probability for any vertex $s \in (t_i, t_{i + 1}]$ it holds that $\deg_t(s) \le \max_{\tau \le t_i} \{\deg_t(\tau)\}$, that is, the \emph{later vertices} (that is, for any $s > t_0$) can have maximum degree only with a negligible probability.
This proof can also be decomposed into three steps: first, we show that a vertex $s$ on its arrival cannot have a degree greater than $(1 + \varepsilon) (p X_t + r)$ with high probability, and then it cannot increase between time $s$ and $t_{i + 1}$ to exceed $X_{t_{i + 1}}$. Finally, to proceed from $\deg_{t_i + 1}(s) \le X_{t_i + 1}$ whp to $\deg_t(s) \le X_t$ whp we use exactly the same Chernoff bound as for early vertices.

To prove the lower bound we follow the steps from the upper bound for the early vertices: we show a respective lower Chernoff-type bound on the growth of $\deg_\tau(s)$ over an interval of certain length $h$ and we combine it with different (but very similar) sequences $t_i$ and $X_{t_i}$, thus proving that in this case $\deg_\tau(s) \ge X_\tau - \ln^{1 + p}(\tau) + 1$ with high probability for all early vertices (that is, $s \le t_0$), and that $X_t \ge (1 - \alpha) t^p$ for any $\alpha > 0$.

Note that the asymmetry between the proofs of both bounds stems from the fact that for the lower bound we only needed to find an inequality that holds with high probability for a single vertex, whereas for the upper bound we had to prove an inequality that holds with high probability for all vertices $s = 1, \ldots, t$.

\subsection[Upper bound, early vertices]{Upper bound, early vertices ($s \le t_0)$}
\label{sec:maximum-upper-early}

We begin with the definitions for two auxiliary sequences that we mentioned earlier:
\begin{definition}
  \label{def:sequences}
  For any $t$ and the given coefficients $\phi(t)$, $(\beta_i(t))_{i = 0}^{k - 1}$ and the sequence of positive jumps $(w_i(t))_{i = 0}^{k - 1}$ we define the sequences $(t_i)_{i = 0}^k$ and $(X_{t_i})_{i = 0}^k$ and a number $k(t) \in \mathbb{N}$, also implicitly dependent on $t$ as follows:
  \begin{align*}
      t_0 & = \phi(t),\qquad t_{i + 1} = t_i + w_i(t), \\
      X_{t_0} & = t_0,\qquad X_{t_{i + 1}} = X_{t_i} + \beta_i(t) \frac{w_i(t) X_{t_i}}{t_i}, \\
      & \text{$k$ is such that $t_{k - 1} < t \le t_k$.}
  \end{align*}
\end{definition}
Moreover, to prove the desired bounds it would be ultimately necessary that $\phi(t)$ and all $w_i(t)$ tend to infinity with $t$.
For brevity, from now on we assume the dependency on $t$ as implicit and write $\phi$, $\beta_i$, and $w_i$ instead of $\phi(t), \beta_i(t), w_i(t)$, respectively.

Note that inductively from the definition it follows that if $\beta_i \le 1$, then $X_{t_i} \le t_i$ for all $i = 0, 1, \ldots, k$.

Moreover, observe that we do not need to specify the values of $X_\tau$ for $\tau$ other than $\{t_0, t_1, \ldots, t_k\}$.
In the rest of the paper we will be using precisely these values in the proofs, so such a definition is sufficient for our purposes.
For reader's convenience we shall assume that for any $\tau \in (t_l, t_{l + 1})$ for some $l = 0, 1, \ldots, k - 1$ the sequence is completed in any way such that $X_{t_l} \le X_\tau \le X_{t_{l + 1}}$.

Now we analyze the asymptotic properties of these sequences.
We start with a simple lower bound:
\begin{lemma}
  \label{lem:lower-2}
  Assume $\beta_i \ge p - \frac{p (1 - p)}{4 \ln{t_i}}$ and $w_i \le \frac{t_i}{\ln{t_i}}$.
  For $t \to \infty$ we have $X_{t_i} \ge t_i^p$ for all $i = 0, 1, \ldots, k$.
\end{lemma}

\begin{proof}
  Let us define $Y_\tau = \tau^p$.
  By definition we know that $X_{t_0} = t_0 \ge Y_{t_0}$.
  Now, let us assume that $X_{t_i} \ge Y_{t_i}$ holds for some $i \ge 0$. Then we have
  \begin{align*}
      Y_{t_{i + 1}} - Y_{t_i} & = \left((t_i + w_i)^p - t_i^p\right)
        = t_i^p \left(\left(1 + \frac{w_i}{t_i}\right)^p - 1\right) \\
      & \le t_i^p \left(\frac{p w_i}{t_i} - \frac{p (1 - p) w_i^2}{4 t_i^2}\right)
        \le t_i^p\,\frac{w_i}{t_i} \left(p - \frac{p (1 - p)}{4 \ln{t_i}}\right).
  \end{align*}
  since from Taylor expansion it follows that $(1 + x)^p \le 1 + p x - \frac{p (1 - p) x^2}{4}$ for any $p \in [0, 1]$ and any $x \in (0, 1)$.
  Therefore,
  \begin{align*}
      Y_{t_{i + 1}} - Y_{t_i}
      & \le Y_{t_i} \frac{w_i}{t_i} \left(p - \frac{p (1 - p)}{4 \ln{t_i}}\right)
        \le X_{t_i} \frac{\beta_i w_i}{t_i}
        = X_{t_{i + 1}} - X_{t_i},
  \end{align*}
  so clearly $X_{t_{i + 1}} \ge Y_{t_{i + 1}}$ holds as well, which completes the inductive step.
\end{proof}

Now we prove an upper bound on $X_t$.

\begin{lemma}
  \label{lem:upper}
  Assume that $\phi \ge \ln{t}$, $\beta_i \le p + \frac{1}{2 \ln{t_i}}$ and $w_i \le \frac{t_i}{\ln{t_i}}$.
  It holds asymptotically as $t \to \infty$ that $X_{t_i} \le \phi^{1 - p} t_i^p \ln{t_i}$ for all $i = 0, 1, \ldots, k$.
\end{lemma}

\begin{proof}
  We again proceed by induction with $Y_\tau = \phi^{1 - p} \tau^p \ln{\tau}$.
  Clearly, $X_{t_0} = t_0 \le Y_{t_0} = t_0 \ln{t_0}$.
  Directly from the definition we get
  \begin{align*}
      Y_{t_{i + 1}} - X_{t_{i + 1}}
        & = Y_{t_{i + 1}} - X_{t_i} \left(1 + \frac{\beta_i w_i}{t_i}\right) \\
        & \ge \phi^{1 - p} t_{i + 1}^p \ln{t_{i + 1}} - \phi^{1 - p} t_i^p \ln{t_i} \left(1 + \frac{\beta_i w_i}{t_i}\right) \\
        & \ge \phi^{1 - p} t_i^p \ln{t_i} \left(\left(\frac{t_{i + 1}}{t_i}\right)^p \left(\frac{\ln{t_{i + 1}}}{\ln{t_i}}\right) - 1 - \frac{\beta_i w_i}{t_i}\right) \\
        & = \phi^{1 - p} t_i^p \ln{t_i} \left(\left(1 + \frac{w_i}{t_i}\right)^p \left(1 + \frac{\ln(1 + w_i / t_i)}{\ln{t_i}}\right) - 1 - \frac{\beta_i w_i}{t_i}\right).
  \end{align*}
  
  Now we use the inequalities derived from the respective Taylor expansions: $(1 + x)^p \ge 1 + p x - \frac{p (1 - p) x^2}{2} \ge 1$ and $\ln(1 + x) \ge x - \frac{x^2}{2} \ge 0$, true for any $p \in [0, 1]$ and any $x \in (0, 1)$. In particular, in our case $x = \frac{w_i}{t_i} \le \frac{1}{\ln{t_i}} \le \frac{1}{\ln{\ln{t}}} = o(1)$. Therefore
  \begin{align*}
      Y_{t_{i + 1}} - X_{t_{i + 1}} & \ge \phi^{1 - p} t_i^p \ln{t_i} \Bigg(\frac{(p - \beta_i) w_i}{t_i} + \left(1 + \frac{p w_i}{t_i}\right) \left(\frac{w_i}{t_i \ln{t_i}} - \frac{w_i^2}{2 t_i^2 \ln{t_i}}\right) \\
        & \qquad - \frac{p (1 - p) w_i^2}{2 t_i^2} \left(1 + \frac{w_i}{t_i \ln{t_i}} - \frac{w_i^2}{2 t_i^2 \ln{t_i}}\right)\Bigg) \\
          & \ge \phi^{1 - p} t_i^{p - 1} \ln{t_i} \cdot w_i \Bigg(- \frac{1}{2 \ln{t_i}} + \frac{1}{\ln{t_i}} - \frac{w_i}{8 t_i} \left(1 + \frac{w_i}{t_i \ln{t_i}} - \frac{w_i^2}{2 t_i^2 \ln{t_i}}\right)\Bigg) \\
          & \ge \phi^{1 - p} t_i^{p - 1} \cdot w_i \left(\frac{3}{8} - \frac{1}{8 \ln{t_i}}\left(\frac{w_i}{t_i} - \frac{w_i^2}{2 t_i^2}\right)\right),
  \end{align*}
  and for sufficiently large $t$ the last expression is clearly non-negative since $\frac{w_i}{t_i} \le \frac{1}{\ln{t_i}} \le \frac{1}{\ln{t_0}} \le \frac{1}{\ln{\phi}} \le \frac{1}{\ln{\ln{t}}} \to 0$, which completes the~proof.
\end{proof}

Next, we need some bounds on $\deg_{\tau}(s)$ holding with high probability to match with the sequence $X_\tau$.
Let us begin with the following estimate:
\begin{lemma}
  \label{lem:simple-growth}
    For any $\phi \le \tau \le t$ and any $0 \le d \le h$ it is true that
    \begin{align*}
        \Pr\!\left[\deg_{\tau + h}(s) - \deg_\tau(s) \ge d~\vert \deg_\tau(s)\right]
            \le \exp\!\left(d \ln{\frac{\exp(1) \cdot h (p \deg_{\tau}(s) + p d + r)}{d \tau}}\right).
    \end{align*}
\end{lemma}

\begin{proof}
    First, it follows from the definition of the model that $\deg_{\tau + i + 1}(s) = \deg_{\tau + i}(s) + I_{\tau + i}$ for $i = 0, 1, \ldots, h - 1$ where $I_{\tau + i} \sim Be(q_{\tau + i})$ for some $q_{\tau + i} \in [0, 1]$. The probability $q_{\tau + i}$ of adding an edge between $s$ and $\tau + i + 1$ is just a sum of probabilities of two events:
    \begin{enumerate}
        \item when $\parent{\tau + i + 1} \in N_{G_{\tau + i}}(s)$ holds, i.e. with probability $\frac{\deg_{\tau + i}(s)}{\tau + i}$ (since we draw the parent uniformly), we add an edge with probability $p$ -- so the whole event has probability $\frac{p \deg_{\tau + i}(s)}{\tau + i}$,
        \item when $\parent{\tau + i + 1} \notin N_{G_{\tau + i}}(s)$ holds, i.e. with probability $1 - \frac{\deg_{\tau + i}(s)}{\tau + i}$, we add an edge with probability $\frac{r}{\tau + i}$ -- so the whole event has probability $\frac{r}{\tau + i} \left(1 - \frac{\deg_{\tau + i}(s)}{\tau + i}\right)$.
    \end{enumerate}
    Both events are disjoint, so we obtain $q_{\tau + i} = \frac{p \deg_{\tau + i}(s) + r}{\tau + i} - \frac{r \deg_{\tau + i}(s)}{(\tau + i)^2} \le \frac{p \deg_{\tau + i}(s) + r}{\tau + i}$.

    Next, we note that the degree grows by at least $d$ if there is a subsequence of $d$ successes $i_1, i_2, \ldots, i_d$ with only failures between them:
    \begin{align*}
        \Pr&\left[\deg_{\tau + h}(s) - \deg_{\tau}(s) \ge d~\vert \deg_\tau(s)\right] \\
            & = \sum_{0 \le i_1 < \ldots < i_d < h} \Pr\!\left[\bigcup_{j \in \{i_1, \ldots, i_d\}} I_{\tau + j} \cup \bigcup_{j \in [0, i_d] \setminus \{i_1, \ldots, i_d\}} \lnot I_{\tau + j}\right] \\
            & = \sum_{0 \le i_1 < \ldots < i_d < h} \prod_{j \in \{i_1, \ldots, i_d\}} \Pr[I_{\tau + j} | \deg_{\tau + j}(s)] \prod_{j \in [0, i_d] \setminus \{i_1, \ldots, i_d\}} \Pr[\lnot I_{\tau + j} | \deg_{\tau + j}(s)].
    \end{align*}
    Now observe that $\Pr[\lnot I_{\tau + j} | \deg_{\tau + j}(s)] \le 1$ for any $j$ and $\Pr[I_{\tau + i_j} | \deg_{\tau + i_j}(s)] \le \frac{p (\deg_\tau(s) + j - 1) + r}{\tau + i_j}$ for $j = 1, 2, \ldots, d$ since $j$-th success occurs after exactly $j - 1$ successes, i.e. when the degree of the vertex $s$ is exactly equal to $\deg_\tau(s) + j - 1$. Thus
    \begin{align*}
        \Pr\!\left[\deg_{\tau + h}(s) - \deg_{\tau}(s) \ge d~\vert \deg_\tau(s)\right]
            & \le \sum_{0 \le i_1 < \ldots < i_d < h} \prod_{j = 1}^{d} \frac{p (\deg_\tau(s) + j - 1) + r}{\tau + i_j} \\
              \le \binom{h}{d} & \max_{0 \le i_1 < \ldots < i_d < h} \left\{\prod_{j = 1}^{d} \frac{p (\deg_\tau(s) + j - 1) + r}{\tau + i_j}\right\}.
    \end{align*}
    One can easily spot that the maximum occurs in the case when $i_j = j - 1$ for all $j = 1, 2, \ldots, d$. This, coupled with a simple upper bound on the value of the binomial coefficient, leads us to the final result
    \begin{align*}
        \Pr\!\left[\deg_{\tau + h}(s) - \deg_{\tau}(s) \ge d~\vert \deg_\tau(s)\right]
            & \le \frac{h^d \exp(d)}{d^d} \prod_{j = 0}^{d - 1} \frac{p (\deg_{\tau}(s) + j) + r}{\tau + j} \\
            \le \exp&\left(d \ln{h} - d \ln{d} + d + d \ln{\frac{p (\deg_{\tau}(s) + d) + r}{\tau}}\right) \\
            \le \exp&\left(d \ln{\frac{\exp(1) \cdot h (p \deg_{\tau}(s) + p d + r)}{d \tau}}\right).
    \end{align*}
\end{proof}
This lemma gives a far better bound than the simple estimation $\deg_{\tau + h}(s) \le \deg_\tau(s) + h$ (e.g. used in \cite{frieze-cocoon}).
However, it is still too coarse to obtain a desired upper bound that could be coupled with the sequence $X_\tau$. But we can still use it to kickstart the Chernoff bound by bounding the probabilities of all Bernoulli variables:
\begin{lemma}
  \label{col:simple-tail}
    For $\ln^{1 + p}{t} \le \tau \le t$, $\varepsilon = \frac{1}{5 \ln{\tau}}$ with $h \le \frac{\varepsilon \tau}{p (1 + 2 \varepsilon) \exp(2)}$ it holds for any constant $A > 0$ that
    \begin{align*}
        \Pr\!\left[\max_{j = 0, \ldots, h - 1} \left\{\frac{p \deg_{\tau + j}(s) + r}{\tau + j}\right\} \ge (1 + \varepsilon) \frac{p X_\tau + r}{\tau}~\Bigg\vert \deg_\tau(s) \le X_\tau\right] = O(t^{-A}).
    \end{align*}
\end{lemma}

\begin{proof}
    Substituting $d = \varepsilon X_\tau$ in \Cref{lem:simple-growth} we get asymptotically as $t \to \infty$ that
    \begin{align*}
        \Pr\Bigg[&\frac{p \deg_{\tau + h}(s) + r}{\tau + h} \ge (1 + \varepsilon) \frac{p X_\tau + r}{\tau}~\Bigg\vert \deg_\tau(s) \le X_\tau\Bigg] \\
            & \le \Pr\!\left[\deg_{\tau + h}(s) \ge (1 + \varepsilon) X_\tau~\vert \deg_\tau(s) \le X_\tau\right] \\
            & \le \Pr\!\left[\deg_{\tau + h}(s) - \deg_{\tau}(s) \ge \varepsilon X_\tau~\vert \deg_\tau(s) \le X_\tau\right] \\
            & \le \exp\!\left(\varepsilon X_\tau \ln{\frac{\exp(1) \cdot h (p X_\tau + p \varepsilon X_\tau + r))}{\varepsilon X_\tau \cdot \tau}}\right) \\
            & \le \exp\!\left(\varepsilon X_\tau \ln{\frac{\exp(1) \cdot h p (1 + 2 \varepsilon) X_{\tau})}{\varepsilon X_\tau \cdot \tau}}\right)
              \le \exp\!\left(- \varepsilon X_\tau\right) \\
            & \le \exp\!\left(- \frac{\max\{\ln^{1 + p}{t}, \tau^p\}}{5 \ln{\tau}}\right)
              \le \exp\!\left(- \frac{\ln{t} \cdot \tau^{p^2 / (1 + p)}\}}{5 \ln{\tau}}\right) \le t^{-A - 1},
    \end{align*}
    for any constant $A > 0$.
    In the fourth line we applied inequality $r \le p \varepsilon X_\tau$.
    Moreover, in the last line we used the facts that $X_\tau \ge \max\{\phi, \tau^p\}$ and $\max\{a, b\} \ge a^\gamma b^{1 - \gamma}$ for any $a, b > 0$ and $\gamma \in [0, 1]$.

    To complete the proof it is sufficient to use a union bound over all values up to $h = O(t)$.
\end{proof}

Let us now proceed with providing a Chernoff-type bound on the growth of the degree of a given early vertex:
\begin{lemma}
  \label{col:chernoff-upper}
  Let $1 \le s \le \tau \le t$ such that $\tau \ge \phi = \ln^{1 + p}{t}$. Then for any $A > 0$ it is true that
  \begin{align*}
       \Pr\!\left[\deg_{\tau + h}(s) - \deg_\tau(s) \ge \frac{3 A (1 + \delta)}{\delta^2} \ln{t}~\Bigg\vert \deg_\tau(s) \le X_\tau\right] = O(t^{-A}),
  \end{align*}
  with $\varepsilon = \delta = \frac{1}{5 \ln{\tau}}$, and $h = \frac{3 A \tau \ln{t}}{\delta^2 (1 + \varepsilon) (p X_\tau + r)}$.
\end{lemma}

\begin{proof}
    Let us first define an event
    \begin{align*}
        \mathcal{D}_\varepsilon(\tau, h) = \left[\max_{j = 0, \ldots, h - 1} \left\{\frac{p \deg_{\tau + j}(s) + r}{\tau + j}\right\} \ge (1 + \varepsilon) \frac{p X_\tau + r}{\tau}~\Bigg\vert \deg_\tau(s) \le X_\tau\right].
    \end{align*}
    Clearly,
    \begin{align*}
       \Pr&\left[\deg_{\tau + h}(s) - \deg_\tau(s) \ge d~\vert \deg_\tau(s) \le X_\tau\right] \\
       & \le \Pr\!\left[\deg_{\tau + h}(s) - \deg_\tau(s) \ge d~\vert \deg_\tau(s) \le X_\tau, \lnot \mathcal{D}_\varepsilon(\tau, h)\right] + \Pr[\mathcal{D}_\varepsilon(\tau, h)],
  \end{align*}

  Let us estimate the probability of the second event. If $h = \frac{3 A \tau \ln{t}}{\delta^2 (1 + \varepsilon) (p X_\tau + r)}$ and $\varepsilon = \frac{1}{5 \ln{\tau}}$, then the condition $h \le \frac{\varepsilon \tau}{p (1 + \varepsilon) \exp(2) \ln{\tau}}$ is met since for some constant $C > 0$ we have
    \begin{align*}
        h & \le \frac{C \tau \ln{t}}{\delta^2 X_\tau} = \frac{C \tau \ln{t} \cdot 25 \ln^2{\tau}}{\max\{\ln^{1 + p}{t}, \tau^p\}} = \frac{C \tau \ln{t} \cdot 25 \ln^2{\tau}}{\ln{t} \cdot \tau^{p^2 / (1 + p)}}
        \le \frac{\tau}{\ln^2{\tau}} \\
        & \le \frac{\tau}{p \cdot 2 \exp(2) \cdot 5 \ln{\tau}} \le \frac{\varepsilon \tau}{p (1 + \varepsilon) \exp(2) \ln{\tau}}
    \end{align*}
    and from \Cref{col:simple-tail} we obtain that $\Pr[\mathcal{D}_\varepsilon(\tau, h)] = O(t^{-A})$.
    Here we again used the facts that $X_\tau \ge \max\{\phi, \tau^p\}$ and $\max\{a, b\} \ge a^\gamma b^{1 - \gamma}$ for any $a, b > 0$ and $\gamma \in [0, 1]$.
    
    Thus, it is sufficient to bound $\deg_{\tau + h}(s) - \deg_\tau(s)$ with high probability when $\mathcal{D}_\varepsilon(\tau, h)$ does not hold, that is, when for all $i = 1, \ldots, h$ it is true that
    \begin{align*}
      \frac{\deg_{\tau + i}(s)}{\tau + i} < (1 + \varepsilon) \frac{X_\tau}{\tau}.
    \end{align*}
    It follows that $I_{\tau + i} = \deg_{\tau + i + 1}(s) - \deg_{\tau + i}(s)$ is stochastically dominated by independent random variables $I^*_{\tau + i} \sim Be\left((1 + \varepsilon) \frac{p X_\tau + r}{\tau}\right)$ for any $i = 0, 1, \ldots, h - 1$ -- since in the case of Bernoulli variables $Be(p_1)$ is stochastically dominated by $Be(p_2)$ whenever $p_1 \le p_2$.
    This way we can eliminate dependencies -- the outcome of each $I_{\tau}$ influences the distributions for $I_{\tau'}$, $\tau' > \tau$ -- and work with independent variables $I^*_{\tau + i}$.

    Now, since the new variables are both Bernoulli and independent, we can use the well-known left tail Chernoff bound for binomial setting from \cite{frieze-introduction} (see Corollary 21.7) which states that for any $\delta \in (0, 1)$
    \begin{align*}
      \Pr\!\left[\sum_{i = 0}^{h - 1} I^*_{\tau + i} \ge (1 + \delta)\,\E\left[\sum_{i = 0}^{h - 1} I^*_{\tau + i}\right]\right] & \le \exp\!\left(- \frac{\delta^2}{3} \E\left[\sum_{i = 0}^{h - 1} I^*_{\tau + i}\right]\right)
    \end{align*}
    and therefore
    \begin{align*}
      \Pr&\left[\deg_{\tau + h}(s) - \deg_\tau(s) \ge (1 + \delta) (1 + \varepsilon)\,\frac{h (p X_\tau + r)}{\tau}~\Bigg\vert \deg_\tau(s) \le X_\tau, \lnot \mathcal{D}_\varepsilon(\tau, h)\right] \\
      & \le \exp\!\left(- \frac{h \delta^2 (1 + \varepsilon) (p X_\tau + r)}{3 \tau}\right).
    \end{align*}
    
    To finish the proof it is sufficient to see that $h = \frac{3 A \tau \ln{t}}{\delta^2 (1 + \varepsilon) (p X_\tau + r)}$ gives the required $O(t^{-A})$ bound in the last equation.
\end{proof}

Finally, we proceed with the proof of the main result of this section.

\begin{theorem}
  \label{thm:upper-bound-early}
  For $G_t \sim \DD(t, p, r)$ with $0 < p < 1$ and $s \in [1, \ln^{1 + p}{t}]$ it holds asymptotically that
  \begin{align*}
    \Pr&\left[\deg_t(s) \ge (1 + \alpha)\,t^p \ln^{2 - p^2}{t}\right] = O(t^{-A})
  \end{align*}
  for any constants $\alpha > 0$ and $A > 0$.
\end{theorem}

\begin{proof}
    Throughout the proof we will use sequences $(t_i)_{i = 0}^k$ and $(X_{t_i})_{i = 0}^k$ with $\phi = \ln^{1 + p}{t}$, $\beta_i = p + \frac{1}{2 \ln{t_i}}$, $w_i = \frac{3 (A + 1) t_i \ln{t}}{\delta^2 (1 + \varepsilon) (p X_{t_i} + r)}$, and $\varepsilon = \delta = \frac{1}{5 \ln{t_i}}$.

    Observe that all the assumptions of \Cref{lem:lower-2} and \Cref{lem:upper} are met, i.e. $p \le \beta_i \le p + \frac{1}{2 \ln{t_i}}$ and $w_i \le \frac{t_i}{\ln{t_i}}$, so we know that $\max\{\ln^{1 + p}{t}, t_i^p\} \le X_{t_i} \le t_i^p \ln^{2 - p^2}{t}$ for all $i = 0, 1, \ldots, k$.

    Now let us define events $\A_i(s) = [\deg_{t_i}(s) < X_{t_i}]$ for $i = 0, \ldots, k$.
    Clearly, $\A_0(s)$ holds since by definition of $X_{t_0}$ we have $\deg_{t_0}(s) < t_0 = X_{t_0}$.

    Suppose that $\A_i(s)$ holds. Then we can apply \Cref{col:chernoff-upper} with $\tau = t_i$ and $h = w_i$:
    \begin{align*}
     \Pr&[\lnot \A_{i + 1}(s) | \A_i(s)] = \Pr[\deg_{t_{i + 1}}(s) \ge X_{t_{i + 1}}~|~\deg_{t_i}(s) < X_{t_i}] \\
       & \le \Pr[\deg_{t_{i + 1}}(s) - \deg_{t_i}(s) \ge X_{t_{i + 1}} - X_{t_i}~|~\deg_{t_i}(s) < X_{t_i}] \\
       & = \Pr\!\left[\deg_{t_{i + 1}}(s) - \deg_{t_i}(s) \ge \beta_i \frac{w_i X_{t_i}}{t_i}~\Big\vert \deg_{t_i}(s) < X_{t_i}\right] \\
       & = \Pr\!\left[\deg_{t_{i + 1}}(s) - \deg_{t_i}(s) \ge \frac{\beta_i X_{t_i}}{(1 + \varepsilon) (p X_{t_i} + r)} \frac{3 (A + 1)}{\delta^2} \ln{t}~\Bigg\vert \deg_{t_i}(s) < X_{t_i}\right] \\
       & \le \Pr\!\left[\deg_{t_{i + 1}}(s) - \deg_{t_i}(s) \ge \frac{3 (A + 1) (1 + \delta)}{\delta^2} \ln{t}~\Bigg\vert \deg_{t_i}(s) < X_{t_i}\right] = O(t^{-A - 1}),
    \end{align*}
    where we used the fact that asymptotically as $t \to \infty$
    \begin{align*}
        \frac{\beta_i X_{t_i}}{(1 + \delta) (1 + \varepsilon) (p X_{t_i} + r)}
            & = \frac{\beta_i X_{t_i}}{X_{t_i} \left(p + \delta \left(p + p \varepsilon + \frac{r (1 + \varepsilon)}{X_{t_i}}\right) + \varepsilon \left(p + \frac{r}{X_{t_i}}\right)\right) + r} \\
            & \ge \frac{\beta_i}{p + \delta \left(p + p \varepsilon + \frac{r (1 + \varepsilon)}{X_{t_i}}\right) + \varepsilon \left(p + p \varepsilon + \frac{r (1 + \varepsilon)}{X_{t_i}}\right) + \frac{r}{X_{t_i}}} \\
            & \ge \frac{p + \frac{1}{2 \ln{t_i}}}{p + \delta + \varepsilon + \frac{r}{X_{t_i}}} \ge 1,
    \end{align*}
    where in the denominator of the first inequality we used the facts that $p + \frac{r}{X_{t_i}} \le p + p \varepsilon + \frac{r (1 + \varepsilon)}{X_{t_i}} = p + o(1) \le 1$ for any constants $0 < p < 1$, $0 \le r \le t_0$ when $t \to \infty$.
    
    Next, we get
    \begin{align*}
      \Pr&[\deg_t(s) \ge X_{t_k}] \le \Pr[\deg_{t_k}(s) \ge X_{t_k}] = \Pr[\lnot \A_k(s)] \\
      & \le \sum_{i = 0}^{k - 1} \Pr[\lnot \A_{i + 1}(s) | \A_i(s)] + \Pr[\lnot \A_0(s)] = \sum_{i = 0}^{k - 1} O(t^{-A - 1}) = O(t^{-A}),
    \end{align*}
    since asymptotically it is true that $w_i \ge 1$ for all $i = 0, \ldots, k$, and therefore $k \le t$.

    To complete the proof it is sufficient to note that $t_k = t_{k - 1} (1 + \alpha) \le (1 + \alpha) t$ and thus $X_{t_k} \le (1 + \alpha) t^p \ln^{2 - p^2}{t}$ for any constant $\alpha > 0$.
\end{proof}

\subsection[Upper bound, late vertices]{Upper bound, late vertices $(s > t_0)$}
\label{sec:maximum-upper-late}

In the second part of the proof we also use the sequences $(t_i)_{i = 0}^k$ and $(X_{t_i})_{i = 0}^k$ as defined in \Cref{def:sequences}.
Moreover,  throughout this section we use the same constants as in the proof of \Cref{thm:upper-bound-early}: $\phi = \ln^{1 + p}{t}$, $\beta_i = p + \frac{1}{2 \ln{t_i}}$ and $w_i = \frac{3 (A + 1) t_i \ln{t}}{\delta^2 (1 + \varepsilon) (p X_{t_i} + r)}$.

The proof consists of showing that for $s \in [t_i, t_{i + 1})$ for some $i = 0, 1, \ldots, k - 1$ the degree graph (i.e. $\deg_s(s)$) is with high probability significantly smaller than its corresponding  $X_{t_{i + 1}}$.
Furthermore, we show that the increase in the degree between $\deg_s(s)$ and $\deg_{t_{i + 1}}(s)$ with high probability cannot compensate for this difference.
Thus, $X_t$ (or, to be more precise, $X_{t_k}$) gives us a good upper bound on $\deg_t(s)$ for all $s$ -- and therefore also we obtain an upper bound for $\Delta(G_t)$.

Let us introduce auxiliary events $\B_l(s) = \bigcup_{\tau = 1}^s \A_l(\tau) = [\deg_{t_l}(\tau)\le X_{t_l}$ for all $\tau \le s \le t_l$] where $\A_i(s)$ is, as before, the event that $\deg_{t_i}(s) \le X_{t_i}$ for a fixed $s \le t_i$.

\begin{lemma}
  \label{lem:start}
  Let $s \in (t_l, t_{l + 1}]$ for some $l = 0, 1, \ldots, k - 1$. Then, for any constants $\varepsilon > 0$ and $A > 0$
  \begin{align*}
    \Pr&\left[\deg_s(s) \ge (1 + \varepsilon) (p X_{t_{l + 1}} + r)~\vert \B_l(t_l) \land \B_{l + 1}(s - 1)\right]
      = O(t^{-A}).
  \end{align*}
\end{lemma}

\begin{proof}
 First, we notice the fact that $\max\{\deg_{t_{l + 1}}(\tau)\colon 1 \le \tau \le s - 1\} \le X_{t_{l + 1}}$ guarantees that $\max\{\deg_s(\tau)\colon 1 \le \tau \le s - 1\} \le X_{t_{l + 1}}$.
 Therefore, $\deg_s(s)$ is stochastically dominated by $A_s \sim Bin(X_{t_{l + 1}}, p) + Bin(s - 1, \frac{r}{s - 1})$ and we directly obtain the result  using the Chernoff bound with $\E[A_s] = p X_{t_{l + 1}} + r$:
  \begin{align*}
    \Pr&\left[\deg_s(s) \ge (1 + \varepsilon) (p X_{t_{l + 1}} + r)~\Big\vert \B_l(t_l) \land \B_{l + 1}(s - 1)\right] \\
     & \le \exp\!\left(- \frac{\varepsilon^2}{\varepsilon + 2} (p X_{t_{l + 1}} + r)\right) \le t^{-A},
  \end{align*}
  asymptotically for any constants $\varepsilon, A > 0$ since $X_{t_{l + 1}} \ge \ln^{1 + p}{t}$.
\end{proof}

Note that the result implies that with high probability at most slightly more than a $p$ fraction of the maximum degree is already present at time $s$.
Therefore, we are interested in bounding the remaining part of the degree, i.e. $\deg_{t_{l + 1}}(s) - \deg_s(s)$, by something smaller than the remaining fraction of the maximum degree.

\begin{lemma}
  \label{lem:jump-1}
  Let $s \in (t_l, t_{l + 1}]$ for some $l = 0, 1, \ldots, k - 1$. Then, for any constant $\alpha > 0$ and $A > 0$
  \begin{align*}
    \Pr&\left[\deg_{t_{l + 1}}(s) - \deg_s(s) \ge \alpha X_{t_{l + 1}}~\vert \B_l(t_l) \land \B_{l + 1}(s - 1)\right] = O(t^{-A}).
  \end{align*}
\end{lemma}

\begin{proof}
    We use \Cref{lem:simple-growth} with $d = \alpha X_{t_{l + 1}}$ to obtain asymptotically as $t \to \infty$ that for any $A > 0$ it holds that
    \begin{align*}
        \Pr&\left[\deg_{t_{l + 1}}(s) - \deg_s(s) \ge \alpha X_{t_{l + 1}}~\vert \B_l(t_l) \land \B_{l + 1}(s - 1)\right] \\
            & = \Pr\!\left[\deg_{t_{l + 1}}(s) - \deg_s(s) \ge \alpha X_{t_{l + 1}}\right]
              \le \Pr\!\left[\deg_{s + w_l}(s) - \deg_s(s) \ge \alpha X_{t_{l + 1}}\right] \\
            & \le \exp\!\left(\alpha X_{t_{l + 1}} \ln{\frac{\exp(1) \cdot w_l p (1 + 2 \alpha) X_{t_{l + 1}}}{\alpha X_{t_{l + 1}} \cdot s}}\right) \\
            & \le \exp\!\left(\alpha X_{t_{l + 1}} \left(\frac{\exp(1) \cdot (1 + 2 \alpha) \cdot 3 (A + 1)}{\alpha (1 + \alpha)} + \ln{\frac{\ln{t}}{\delta^2 (X_{t_l} + r / p)}}\right)\right) \\
            & \le \exp\!\left(\alpha X_{t_{l + 1}} \left(\Theta(1) + \ln{\frac{25 \ln{t} \cdot \ln^2{t_l}}{\max\{\ln^{1 + p}{t}, t_l^p\}}}\right)\right) \\
            & \le \exp\!\left(\alpha \ln^{1 + p}{t} \left(\Theta(1) + \ln{\frac{25 \ln^2{t_l}}{t_l^{p^2 / (1 + p)}}}\right)\right)
              \le \exp\!\left(-A \ln{t}\right) \le t^{-A}
    \end{align*}
    as needed.
\end{proof}

To proceed we need the following two lemmas.

\begin{lemma}
  \label{lem:jump-2}
  Let $s \in (t_l, t_{l + 1}]$ for some $l = 0, 1, \ldots, k - 1$.
  Then asymptotically as $t \to \infty$, for any constant $A > 0$ it holds that
  \begin{align*}
    \Pr&\left[\deg_{t_{l + 1}}(s) \ge X_{t_{l + 1}} | \B_l(t_l) \land \B_{l + 1}(s - 1)\right] = O(t^{-A}).
  \end{align*}
\end{lemma}

\begin{proof}
  We combine \Cref{lem:start} with $\varepsilon = \frac{1 - p}{4 p}$ and \Cref{lem:jump-1} with $\alpha = \frac{1 - p}{2}$ to obtain
  \begin{align*}
      \Pr&\left[\deg_{t_{l + 1}}(s) \ge X_{t_{l + 1}}~\vert \B_l(t_l) \land \B_{l + 1}(s - 1)\right] \\
        & \le \Pr\!\left[\deg_s(s) \ge \left(1 + \frac{1 - p}{4 p}\right) (p X_{t_{l + 1}} + r)~\Big\vert \B_l(t_l) \land \B_{l + 1}(s - 1)\right] \\
        & \quad + \Pr\!\left[\deg_{t_{l + 1}}(s) - \deg_s(s) \ge \frac{1 - p}{2} X_{t_{l + 1}}~\Big\vert \B_l(t_l) \land \B_{l + 1}(s - 1)\right] = O(t^{-A}).
  \end{align*}
\end{proof}

\begin{lemma}
  \label{lem:jump-3}
  Let $s \in (t_l, t_{l + 1}]$ for some $l = 0, 1, \ldots, k - 1$.
  Then asymptotically as $t \to \infty$, for any constant $A > 0$ it holds that
  \begin{align*}
    \Pr&\left[\lnot\B_{l + 1}(t_{l + 1}) | \B_l(t_l)\right] = O(t^{-A}).
  \end{align*}
\end{lemma}

\begin{proof}
  Let $l$ be the first value for which the lemma does not hold.
  Then, from \Cref{lem:jump-2} we get that for any constant $A > 0$ it holds that
  \begin{align*}
    \Pr&\left[\lnot\B_{l + 1}(t_{l + 1}) | \B_l(t_l) \land \B_{l + 1}(t_l)\right]
        = \sum_{s = t_l}^{t_{l + 1} - 1} \Pr\!\left[\lnot\B_{l + 1}(s + 1) | \B_l(t_l) \land \B_{l + 1}(s)\right] \\
      & = \sum_{s = t_l}^{t_{l + 1} - 1} \Pr\!\left[\lnot\A_{l + 1}(s + 1) | \B_l(t_l) \land \B_{l + 1}(s)\right]
        = O(t^{-A}).
  \end{align*}
  
  From \Cref{thm:upper-bound-early} we know that $\Pr\!\left[\B_0(t_0)\right] = 1 - O(t^{-A})$.
  Recall that by our assumption $\Pr\!\left[\lnot\B_{i + 1}(t_{i + 1}) | \B_i(t_i)\right] = 1 - O(t^{-A})$ for all $i = 0, 1, \ldots, l - 1$, so it follows that $\Pr\!\left[\B_i(t_i)\right] = 1 - O(t^{-A})$ for all $i = 0, 1, \ldots, l$.
  We use this fact, combined with the observation that $\B_l(t_l) \subseteq \A_l(s)$ and \Cref{thm:upper-bound-early} to get
  \begin{align*}
    \Pr&\left[\lnot \B_{l + 1}(t_l) | \B_l(t_l)\right]
        \le \sum_{s = 1}^{t_l} \Pr\!\left[\lnot \A_{l + 1}(s) | \B_l(t_l)\right] \\
        & \le \sum_{s = 1}^{t_l} \frac{\Pr\!\left[\lnot \A_{l + 1}(s) \land \B_l(t_l)\right]}{\Pr\!\left[\B_l(t_l)\right]}
          \le \sum_{s = 1}^{t_l} \frac{\Pr\!\left[\lnot \A_{l + 1}(s) \land \A_l(s)\right]}{\Pr\!\left[\B_l(t_l)\right]} \\
        & \le \sum_{s = 1}^{t_l} \frac{\Pr\!\left[\lnot \A_{l + 1}(s) | \A_l(s)\right]}{\Pr\!\left[\B_l(t_l)\right]}
        = \sum_{s = 1}^{t_l} \frac{O(t^{-A})}{1 - O(t^{-A})} = O(t^{-A}).
  \end{align*}
  
  Finally, for any events $E_1$, $E_2$, $E_3$ we have
  \begin{align*}
      \Pr[\lnot E_1 | E_2]
        & = \Pr[\lnot E_1 \land E_3 | E_2] + \Pr[\lnot E_1 \land \lnot E_3 | E_2] \\
        & \le \Pr[\lnot E_1 | E_3 \land E_2] + \Pr[\lnot E_3 | E_2].
  \end{align*}
  We substitute $E_1 = \B_{l + 1}(t_{l + 1})$, $E_2 = \B_l(t_l)$ and $E_3 = \B_{l + 1}(t_l)$ to obtain the final result.
\end{proof}

Finally, we present the main result of this section.

\begin{theorem}
  \label{thm:jump}
  For $G_t \sim \DD(t, p, r)$ with $0 < p < 1$ and any constants $\alpha, A > 0$ it holds asymptotically that
  \begin{align*}
    \Pr&\left[\Delta(G_t) \ge (1 + \alpha) t^p \ln^{2 - p^2}{t}\right] = O(t^{-A}).
  \end{align*}
\end{theorem}

\begin{proof}
  From \Cref{lem:upper} we know that $X_{t_k} \le (1 + \alpha) t^p \ln^{2 - p^2}{t}$ holds asymptotically. It follows that in this case
  \begin{align*}
    \Pr\!\left[\Delta(G_t) \ge (1 + \alpha) t^p \ln^{2 - p^2}{t}\right]
    & \le \Pr\!\left[\Delta(G_t) \ge X_{t_k}\right]
      \le \Pr\!\left[\lnot\B_k(t_k)\right] \\
    & \le \sum_{l = 0}^{k - 1} \Pr\!\left[\lnot\B_{l + 1}(t_{l + 1}) | \B_l(t_l)\right] + \Pr\!\left[\lnot\B_0(t_0)\right].
  \end{align*}
  
  Now, from \Cref{thm:upper-bound-early} and \Cref{lem:jump-3} we know that both $\Pr\!\left[\lnot\B_0(t_0)\right] = O(t^{-A})$ and $\Pr[\lnot\B_{l + 1}(t_l) | \B_l(t_l)] = O(t^{-A})$ for any $A > 0$, respectively.
  Putting this all together with the fact that asymptotically as $t \to \infty$ it holds that $k \le t$ we obtain the final result.
\end{proof}

\subsection{Lower bound}
\label{sec:maximum-lower}

Here we proceed analogously to the case of the upper bound for early vertices.
We provide an appropriate Chernoff-type bound for the degree of a given vertex with respect to some deterministic sequence. Then we again use a special sequence, which has the desired rate of growth and serves as a lower bound on $\deg_t(s)$.
Note that we don't need to extend our analysis for the late vertices since a lower bound for the degree of any vertex $s$ at time $t$ is also a lower bound for the minimum degree of $G_t$.

Now, we note that if we start the whole process from a non-empty graph, then there exists $s \in [1, t_0]$ such that $\deg_{t_0}(s) \ge 1$.
Moreover, even if the starting graph is empty, but $r > 0$, then with high probability there exists a vertex with positive degree, as the probability of adding another isolated vertex to an empty graph on $t$ vertices is at most $(1 - \frac{r}{t})^t \le \exp(-r)$, so within first $\frac{A}{r} \ln{t}$ vertices for any $A > 0$ we have a non-isolated vertex with probability at least $1 - O(t^{-A})$.
Of course, if we start from an empty graph and $r = 0$, then for any $p$ there is no edge in the duplication process. However, in this case it trivially follows that $\Delta(G_t) = 0$, so we omit this case in further analysis.

That said, let us now proceed with the aforementioned Chernoff-type lower bound for the degree of a given early vertex:
\begin{lemma}
  \label{lem:chernoff-lower}
  Let $1 \le s \le \tau \le t$ such that $\tau \ge \phi = \ln^{1 + p}{t}$. Then for any $A > 0$ it is true that
  \begin{align*}
       \Pr\!\left[\deg_{\tau + h}(s) - \deg_\tau(s) \le \frac{2 A (1 - \delta)}{\delta^2} \ln{t}~\Bigg\vert \deg_\tau(s) \ge X_\tau\right] = O(t^{-A}),
  \end{align*}
  with $\varepsilon = \delta = \frac{p (1 - p)}{8 \ln{\tau}}$ and $h = \frac{2 A \tau \ln{t}}{\delta^2 (1 - \varepsilon) (p X_\tau + r)}$.
\end{lemma}

\begin{proof}
 Let us recall (as in the proof of \Cref{lem:simple-growth}) that for $i = 0, 1, \ldots, h - 1$ we have $\deg_{\tau + i + 1}(s) = \deg_{\tau + i}(s) + I_{\tau + i}$ where $I_{\tau + i} \sim Be\left(q_{\tau + i}\right)$ for $q_{\tau + i} = \frac{p \deg_{\tau + i}(s) + r}{\tau + i} - \frac{r \deg_{\tau + i}(s)}{(\tau + i)^2}$.
  Also clearly $\deg_{\tau + i}(s) \ge \deg_\tau(s)$ for any $i = 0, 1, \ldots, h$, so we have
  \begin{align*}
        q_{\tau + i} = & \frac{p \deg_{\tau + i}(s) \left(1 - \frac{r}{p (\tau + i)}\right) + r}{\tau + i}
            \ge \frac{p \deg_{\tau}(s) \left(1 - \frac{r}{p \tau}\right) + r}{\tau + h} \\
        & \ge \frac{p X_\tau \left(1 - \varepsilon^2\right) + r}{\tau (1 + \varepsilon)} \ge (1 - \varepsilon) \frac{p X_\tau + r}{\tau},
  \end{align*}
  since for $\varepsilon = \frac{p (1 - p)}{8 \ln{\tau}}$ it holds that $h \le \varepsilon t$ and $\varepsilon^2 \ge \frac{r}{p \tau}$.
  Therefore for any $i = 0, 1, \ldots, h - 1$ we know that $I_{\tau + i}$ stochastically dominates $I^*_{\tau + i} \sim Be\left((1 - \varepsilon) \frac{p X_\tau + r}{\tau}\right)$.

  As in the proof of the upper bound, the new variables are both Bernoulli and independent. So this time we can use the right tail Chernoff bound for binomial setting from \cite{frieze-introduction} (see Corollary 21.7) which states that for any $\delta \in (0, 1)$
  \begin{align*}
      \Pr\!\left[\sum_{i = 0}^{h - 1} I^*_{\tau + i} \le (1 - \delta)\,\E\left[\sum_{i = 0}^{h - 1} I^*_{\tau + i}\right]\right] & \le \exp\!\left(- \frac{\delta^2}{2} \E\left[\sum_{i = 0}^{h - 1} I^*_{\tau + i}\right]\right)
  \end{align*}
  and therefore
  \begin{align*}
      \Pr\!\left[\deg_{\tau + h}(s)\!\le\!\deg_\tau(s) + (1 - \delta) (1 - \varepsilon)\,\frac{h (p X_\tau + r)}{\tau}\right]\!\le\!\exp\!\left(- \frac{h \delta^2 (1 - \varepsilon) (p X_\tau + r)}{2 \tau}\right)
  \end{align*}
  as clearly $\Pr\!\left[\deg_{\tau + h}(s) - \deg_\tau(s) \le k\right] = \Pr\!\left[\sum_{i = 0}^{h - 1} I_{\tau + i} \le k\right] \le \Pr\!\left[\sum_{i = 0}^{h - 1} I^*_{\tau + i} \le k\right]$ for any $k$, due to the stochastic dominance.
    
  To finish the proof it is sufficient to see that $h = \frac{2 A \tau \ln{t}}{\delta^2 (1 - \varepsilon) (p X_\tau + r)}$ gives the required $O(t^{-A})$ bound in the last equation.
\end{proof}

In the following, we again use sequences $(t_i)_{i = 1}^k$ and $(X_{t_i})_{i = 1}^k$ from \Cref{def:sequences}.
Let us also define $\C_i(s) = [\deg_{t_i}(s) > X_{t_i} - \phi + 1]$ for a fixed $s \le t_i$. Now we are in the position to proceed with the main theorem of this section:
\begin{theorem}
  \label{thm:lower-bound-early}
  For $G_t \sim \DD(t, p, r)$ with $0 < p < 1$ there exists $s$ such that it holds asymptotically that
  \begin{align*}
    \Pr&\left[\deg_t(s) < (1 - \alpha) t^p\right] = O(t^{-A})
  \end{align*}
  for any constants $\alpha, A > 0$.
\end{theorem}

\begin{proof}
    Let us use $\phi = \ln^{1 + p}{t}$, $\beta_i = p - \frac{p (1 - p)}{4 \ln{t_i}}$ and $w_i = \frac{2 (A + 1) t_i \ln{t}}{\delta^2 (1 - \varepsilon) (p X_{t_i} + r)}$ with $\delta = \varepsilon = \frac{p (1 - p)}{8 \ln{t_i}}$.

    Suppose that $\C_i(s)$ holds. Then we can apply \Cref{lem:chernoff-lower} with $\tau = t_i$ and $h = w_i$:
    \begin{align*}
     \Pr&[\lnot \C_{i + 1}(s) | \C_i(s)] = \Pr[\deg_{t_{i + 1}}(s) \le X_{t_{i + 1}}~|~\deg_{t_i}(s) > X_{t_i} - \phi + 1] \\
       & \le \Pr[\deg_{t_{i + 1}}(s) - \deg_{t_i}(s) \le X_{t_{i + 1}} - X_{t_i}~|~\deg_{t_i}(s) > X_{t_i} - \phi + 1] \\
       & = \Pr\!\left[\deg_{t_{i + 1}}(s) - \deg_{t_i}(s) \le \beta_i \frac{w_i X_{t_i}}{t_i}~\Big\vert \deg_{t_i}(s) > X_{t_i} - \phi + 1\right] \\
       & \le \Pr\!\left[\deg_{t_{i + 1}}(s) - \deg_{t_i}(s) \le \frac{2 (A + 1) (1 - \delta)}{\delta^2} \ln{t}~\Bigg\vert \deg_{t_i}(s) > X_{t_i} - \phi + 1\right] \\
       & = O(t^{-A - 1}),
    \end{align*}
    where we used the fact that asymptotically as $t \to \infty$ it holds that
    \begin{align*}
        \frac{\beta_i X_{t_i}}{(1 - \delta) (1 - \varepsilon) (p X_{t_i} + r)}
            & \le \frac{p - \frac{p (1 - p)}{4 \ln{t_i}}}{p (1 - \delta - \varepsilon)} = 1.
    \end{align*}
    
    Next, we get
    \begin{align*}
      \Pr&[\deg_t(s) \le X_{t_k} - \phi + 1] \le \Pr[\deg_{t_k}(s) \le X_{t_k} - \phi + 1] = \Pr[\lnot \C_k(s)] \\
      & \le \sum_{i = 0}^{k - 1} \Pr[\lnot \C_{i + 1}(s) | \C_i(s)] + \Pr[\lnot \C_0(s)] = \sum_{i = 0}^{k - 1} O(t^{-A - 1}) = O(t^{-A}),
    \end{align*}
    since asymptotically it is true that $w_i \ge 1$ for all $i = 0, \ldots, k$, and therefore $k \le t$.

    To complete the proof it is sufficient to note that $t \le t_k \le (1 + \alpha) t_{k - 1} \le (1 + \alpha) t$ for any constant $\alpha > 0$ and thus $X_{t_k} \le (1 + \alpha) t^p$.
\end{proof}

\section{Average degree}
\label{sec:average}

Now let us proceed to the results on the average degree of $G_t$ defined as
\begin{align*}
    D(G_t)= \frac{1}{t} \sum_{s=1}^t \deg_t(s).
\end{align*}
    
First, we recall from \cite[Theorem 9(iii)]{turowski-expected} that for any $\tau = t_0, \ldots, t - 1$ it holds asymptotically (i.e. when $t_0 \to \infty$) that
\begin{align*}
\E[\deg_{\tau}(\tau)] & =
\begin{cases}
    D(G_{t_0}) \frac{p \Gamma(t_0) \Gamma(t_0 + 1)}{\Gamma(t_0 + c_3) \Gamma(t_0 + c_4)} \tau^{2 p - 1} (1 + o(1)) & \text{if $p \le \frac{1}{2}$, $r = 0$}, \\
    r (1 + o(1)) & \text{if $p = 0$, $r > 0$}, \\
    \left(\frac{r (1 - p)}{p (1 - 2 p)} - \frac{r}{p}\right) (1 + o(1)) & \text{if $0 < p < \frac{1}{2}$, $r > 0$}, \\
    r \log{\tau}\, (1 + o(1)) & \text{if $p = \frac{1}{2}$, $r > 0$}, \\
    \left(D(G_{t_0}) + \frac{2 r t_0}{t_0^2 + 2 p t_0 - 2 r}\, \Hypergeometric{3}{2}{t_0 + 1, t_0 + 1, 1}{t_0 + c_3 + 1, t_0 + c_4 + 1}{1}\right) \\
        \qquad \frac{p \Gamma(t_0) \Gamma(t_0 + 1)}{\Gamma(t_0 + c_3) \Gamma(t_0 + c_4)} \tau^{2 p - 1} (1 + o(1)) & \text{if $p > \frac{1}{2}$},
\end{cases}
\end{align*}
where $D(G_{t_0})$ is the average degree of the initial graph $G_{t_0}$ and
\begin{equation*}
  \Hypergeometric{3}{2}{a_1, a_2, a_3}{b_1, b_2}{z}
    = \sum_{l = 0}^\infty \frac{(a_1)_l (a_2)_l (a_3)_l}{(b_1)_l (b_2)_l} \frac{z^l}{l!}
\end{equation*}
is the generalized hypergeometric function with $(a)_l = a (a + 1) \ldots (a + l - 1)$, $(a)_0 = 1$
the rising factorial (see \cite{abramowitz} for details).

In short, if we omit constant factors, there are three regimes of growth: constant, $\ln{t}$, and $t^{2 p - 1}$.
We need to find the proper high probability bound for each case separately, however it turns out that the proofs are very similar.

\subsection{Upper bound}

Now we may proceed to the main result of this section: the upper bound for the average degree of $G_t$. It turns out that there are exactly two regimes with somewhat different behavior:
\begin{theorem}
\label{thm:ub_avg_deg}
Asymptotically for $G_t \sim \DD(t, p, r)$ it holds that
\begin{align*}
  \Pr&[D(G_t) \ge A\, C\, \ln{t}] = O(t^{-A})           & \text{for $p \le \frac{1}{2}$,} \\
  \Pr&[D(G_t) \ge C\, t^{2 p - 1}] = O(t^{-A})          & \text{for $p > \frac{1}{2}$.}
\end{align*}
for some fixed constant $C > 0$ and any $A > 0$.
\end{theorem}

\begin{proof}
    For simplicity, we will work with the total number of edges $\tau D(G_\tau)$ instead of $D(G_\tau)$.
    Clearly, for any $\tau = t_0, \ldots, t - 1$ it holds that
    \begin{align*}
        (\tau + 1) D(G_{\tau + 1}) - \tau D(G_\tau) & = 2 \deg_{\tau + 1}(\tau + 1), \\
        \deg_{\tau + 1}(\tau + 1) \sim Bin(\deg_\tau(\parent{\tau + 1}), p) & + Bin(\tau - \deg_\tau(\parent{\tau + 1}), r / \tau).
    \end{align*}
    Therefore, we can use Chernoff bound to obtain for any $\delta \ge 0$
    \begin{align*}
        \Pr&\left[(\tau + 1) D(G_{\tau + 1}) - \tau D(G_\tau) \ge 2\,(1 + \delta)\,\E[\deg_{\tau + 1}(\tau + 1)]\right] \\
              & \leq \exp\!\left(- \frac{2 \delta^2}{2 + \delta} \E[\deg_{\tau + 1}(\tau + 1)]\right).
    \end{align*}

    Now, for $p > \frac{1}{2}$ we know that $\E[\deg_{\tau}(\tau)] \le C^* \tau^{2 p - 1}$ for some constant $C^* > 0$.
    Thus, it is sufficient to set $t_0 = t^{p/3}$ and $\delta = \sqrt{\frac{3 (A + 1) \ln{t}}{2 C^* \tau^{2 p - 1}}} = o(1)$ for all $\tau = t_0, \ldots, t - 1$ to get
    \begin{align*}
        \Pr&\left[(\tau + 1) D(G_{\tau + 1}) - \tau D(G_\tau) \ge 2 (1 + \delta) C^*\,\tau^{2 p - 1}\right] = O(t^{-A - 1}),
    \end{align*}
    and by summing over all $\tau$ that no event from polynomial tails happens we obtain
    \begin{align*}
        \Pr&\left[t D(G_t) \ge C\,t^{2 p}\right] \le \Pr\!\left[t D(G_t) - t_0 D(G_{t_0}) \ge \sum_{i = t_0}^{t - 1} 2 (1 + \delta) C^*\,\tau^{2 p - 1}\right] = O(t^{-A}),
    \end{align*}
    for any constant $C \ge t^{- 2 p} \sum_{i = t_0}^{t - 1} 2 (1 + \delta) C^*\,\tau^{2 p - 1} + t^{- 2 p} t_0 D(G_{t_0})$ -- and such constant indeed exists since it is not hard to verify that the latter sum is finite.

    In all cases $0 < p \le \frac{1}{2}$ it turns out that $\sqrt{\frac{3 (A + 1) \ln{t}}{2 C^* \tau^{2 p - 1}}} \to \infty$.
    However, for $0 < p \le \frac{1}{2}$, $r > 0$ we have $\E[\deg_{\tau}(\tau)] \le C^* \ln{\tau}$ for some constant $C^* > 0$, and we can assume $\delta \to \infty$ such that
    \begin{align*}
        \frac{1 + \delta}{2} \le \frac{\delta^2}{2 + \delta} = \frac{(A + 1) \ln{t}}{2 C^* \ln{\tau}},
    \end{align*}
    so therefore
    \begin{align*}
        \Pr\!\left[(\tau + 1) D(G_{\tau + 1}) - \tau D(G_\tau) \ge 2 (A + 1) \ln{t}\right] = O(t^{-A - 1}), \\
        \Pr\!\left[t D(G_t) \ge A\,C\,t \ln{t}\right] \le \Pr\!\left[t D(G_t) - t_0 D(G_{t_0}) \ge \sum_{i = t_0}^{t - 1} 2 (A + 1) \ln{i}\right] = O(t^{-A}),
    \end{align*}
    for some constant $C \ge 2 + \frac{t_0}{A t \ln{t}} D(G_{t_0})$ when $t_0 = t^{1/3}$.

    Finally, let us study the case $0 < p < \frac{1}{2}$, $r = 0$. Again we know that $\E[\deg_{\tau}(\tau)] \le C^* \tau^{2 p - 1}$ for some constant $C^* > 0$.
    Again, we can assume
    \begin{align*}
        \frac{1 + \delta}{2} \le \frac{\delta^2}{2 + \delta} = \frac{(A + 1) \ln{t}}{2 C^* \tau^{2 p - 1}},
    \end{align*}
    so by a similar reasoning as before we get
    \begin{align*}
        \Pr\!\left[(\tau + 1) D(G_{\tau + 1}) - \tau D(G_\tau) \ge 2 (A + 1) \ln{t}\right] & = O(t^{-A - 1}), \\
        \Pr\!\left[t D(G_t) \ge A\,C\,t \ln{t}\right] \le \Pr\!\left[t D(G_t) - t_0 D(G_{t_0}) \ge \sum_{i = t_0}^{t - 1} 2 (A + 1) \ln{t}\right] & = O(t^{-A}),
    \end{align*}
    for sufficiently large constant $C$ when $t_0 = t^{1/3}$. 
\end{proof}

\subsection{Lower bound}

We now turn our attention to establishing the corresponding lower bound. Note that since $\E[D(G_{t})] = O(\log{t})$ for $p \le \frac{1}{2}$,
the lower polynomial tail is trivial in this range since all smaller values are within the polylogarithmic distance from the mean.
However, we can investigate the case $p > \frac{1}{2}$.

\begin{theorem}
\label{thm:lb_avg_deg}
For $G_t \sim \DD(t, p, r)$ with $p > \frac{1}{2}$ asymptotically it holds that
\begin{align*}
  \Pr\!\left[D(G_{t}) \le C\,t^{2 p - 1}\right] = O(t^{-A}).
\end{align*}
for some fixed constant $C > 0$ and any $A > 0$.
\end{theorem}

\begin{proof}
    Similarly as before, we invoke the appropriate Chernoff bound for $\delta \in (0, 1)$
    \begin{align*}
        \Pr&\left[(\tau + 1) D(G_{\tau + 1}) - \tau D(G_\tau) \le 2\,(1 - \delta)\,\E[\deg_{\tau + 1}(\tau + 1)]\right] \\
              & \leq \exp\!\left(-\delta^2 \E[\deg_{\tau + 1}(\tau + 1)]\right).
    \end{align*}

    For $p > \frac{1}{2}$ it is true that $\E[\deg_{\tau}(\tau)] \ge C^* \tau^{2 p - 1}$ for some constant $C^* > 0$.
    Thus, it is sufficient to set $t_0 = t^{p/3}$ and $\delta = \sqrt{\frac{(A + 1) \ln{t}}{C^* \tau^{2 p - 1}}} \le \frac{1}{2}$ for all $\tau = t_0, \ldots, t - 1$ to get
    \begin{align*}
        \Pr&\left[(\tau + 1) D(G_{\tau + 1}) - \tau D(G_\tau) \le 2 (1 - \delta)\,C^* \tau^{2 p - 1}\right] = O(t^{-A - 1}),
    \end{align*}
    which leads us to
    \begin{align*}
        \Pr&\left[t D(G_t) - t_0 D(G_{t_0}) \le C\,t^{2 p}\right] \\
            & \le \Pr\!\left[t D(G_t) - t_0 D(G_{t_0}) \le \sum_{i = t_0}^{t - 1} 2 (1 - \delta) C^*\,\tau^{2 p - 1}\right] = O(t^{-A}).
    \end{align*}
    for any constant $0 < C \le t^{- 2 p} \sum_{i = t_0}^{t - 1} 2 (1 - \delta) C^*\,\tau^{2 p - 1} + t^{- 2 p} t_0 D(G_{t_0})$ -- and such constant indeed exists since it is not hard to verify that the latter sum is non-zero and finite when $t_0 = t^{1/3}$.
\end{proof}
  
\section{Further challenges}
In this paper we focus on deriving large deviations for the average and the maximum degree in the duplication-divergence networks.
By a simple martingale argument one can show that $\Delta(G_t) / t^p$ converges to some random variable $\Delta$.
However, it is still worth asking whether $\Delta$ has finite support (e.g. dependent only on $p$ and $r$, but not on $t$).

A natural next challenge would be to obtain the exact asymptotic formula for the whole degree distribution.
For example, there is an open question whether $DD(t, p, r)$ graphs are scale-free, i.e. they have $\Theta(k^{-\gamma})$ fraction of vertices with degree $k$.
A first step towards this goal was already done for $r = 0$ in \cite{jordan,jacquet-aofa}, where it was proved that this property indeed holds for the (only) giant component $p < e^{-1}$. However, it was noticed in \cite{hermann} that for $r = 0$ and all $0 < p < 1$ such phenomenon does not appear in the whole graph, since almost all vertices are isolated, thus for any $k > 0$ the fraction of vertices of degree $k$ tends to $0$ as $t \to \infty$.

Finally, finding good bounds on the concentration of both $D(G_t)$ and $\Delta(G_t)$ is only the step towards the full understanding of this model, as we still do not know for example how symmetric such networks are. This, in turn, we believe could help find good compression algorithms for these types of networks, as was the case with other graph models \cite{chierichetti-compressible,luczak-isit}. 

\bibliographystyle{siamplain}
\bibliography{concentration}
\end{document}